\documentclass[journal,letterpaper,twocolumn,twoside]{IEEEtran}
\usepackage{comment}

\usepackage[utf8]{inputenc} 
\usepackage[T1]{fontenc}
\usepackage{url}
\usepackage{ifthen}
\usepackage{cite}
\usepackage{amsmath}%
\usepackage{wasysym}%
\usepackage{amssymb}

\usepackage{mdwmath}
\usepackage{blindtext}
\usepackage{eqparbox}
\usepackage{lipsum}
\usepackage{xcolor}
\usepackage{algorithm}
\usepackage{enumerate,multirow}
\usepackage{algpseudocode}

\usepackage{times,amssymb,amsmath,amsfonts,nicefrac,color,bbm,mathrsfs,float}
\usepackage{multirow,tikz,graphicx}

\usepackage{xcolor}
\usepackage{comment}
\definecolor{MyGreen1}{RGB}{20,180,40}
\definecolor{MyBlue1}{RGB}{00,150,255}
\definecolor{MyGray1}{RGB}{200,200,200}

\usetikzlibrary{shapes.misc, positioning}
\usetikzlibrary{decorations.pathreplacing}

\usepackage{makecell}
\usetikzlibrary{fit}

\usetikzlibrary{arrows.meta,
	chains,
	decorations.pathreplacing,calligraphy,
	positioning,
	quotes,
	shapes.geometric,
	decorations.text
}

\tikzset{
	BC/.style = {decorate,  
		decoration={calligraphic brace, amplitude=1.2mm,
			raise=1.2mm, mirror},
		thick, pen colour={black}
	},
}

\tikzset{
	BC2/.style = {decorate,  
		decoration={calligraphic brace, amplitude=2mm,
			raise=1.5mm, mirror},
		thick, pen colour={black}
	},
}

\tikzset{
	block/.style    = {draw, thick, rectangle, text centered, align=center,  minimum width = 2em},
	sblock/.style      = {draw, thick, rectangle, minimum height = 2em,
		minimum width = 2em}, 
}

\tikzstyle{block} = [rectangle, draw,  text centered]

\usepackage{siunitx}

\usepackage{calc}
\newcolumntype{M}[1]{>{\centering\arraybackslash}m{#1}}
\newcolumntype{N}{@{}m{0pt}@{}}

\usepackage{amsthm}

\newtheorem{theorem}{{Theorem}}
\newtheorem{lemma}[theorem]{{Lemma}}

\DeclareMathAlphabet{\mathbfsl}{OT1}{ppl}{b}{it} 

\newcommand{\be}[1]{\begin{equation}\label{#1}}
\newcommand{\ee}{\end{equation}}

\renewcommand{\leq}{\leqslant}

\newcommand{\Lref}[1]{Lem\-ma\,\ref{#1}}
\newcommand{\Cref}[1]{Co\-ro\-lla\-ry\,\ref{#1}}

\IEEEoverridecommandlockouts
\begin{document}
\title{Low-Complexity Decoding of a Class of Reed-Muller Subcodes for Low-Capacity Channels}  
\author{ \IEEEauthorblockN{Mohammad Vahid Jamali$^{\ast}$, Mohammad Fereydounian$^{\dagger}$, Hessam Mahdavifar$^{\ast}$, and Hamed Hassani$^{\dagger}$}
\IEEEauthorblockA{$^{\ast}$Department of Electrical Engineering and Computer Science, University of Michigan, \{mvjamali,hessam\}@umich.edu\\
	$^{\dagger}$Department of Electrical and Systems Engineering, University of Pennsylvania, \{mferey,hassani\}@seas.upenn.edu}
}	
\maketitle
\begin{abstract}
We present a low-complexity and low-latency decoding algorithm for a class of Reed-Muller (RM) subcodes that are defined based on the product of smaller RM codes. More specifically, the input sequence is shaped as a multi-dimensional array, and the encoding over each dimension is done separately via a smaller RM encoder. Similarly, the decoding is performed over each dimension via a low-complexity decoder for smaller RM codes.
The proposed construction is of particular interest to low-capacity channels that are relevant to emerging low-rate communication scenarios. We present an efficient soft-input soft-output (SISO) iterative decoding algorithm for the product of RM codes and demonstrate its superiority compared to hard decoding over RM code components. The proposed coding scheme has decoding (as well as encoding) complexity of $\mathcal{O}(n\log n)$ and latency of $\mathcal{O}(\log n)$ for blocklength $n$. This research renders a general framework toward efficient decoding of RM codes. 

\end{abstract}
 \thispagestyle{empty}
\section{Introduction}\label{intro}
In recent years, there has been significant renewed interest in exploring Reed-Muller (RM) codes, which are one of the oldest families of error-correcting codes \cite{reed1954class,muller1954application}. RM codes are closely connected to polar codes \cite{arikan2009channel} in the sense that the generator matrices of both codes are obtained by selecting rows from a same matrix, though by different selection rules. In contract to polar codes, which have channel-specific construction, RM codes have a universal encoding scheme.
It is also conjectured that RM codes have similar characteristics to random codes in terms of weight enumeration \cite{kaufman2012weight} and scaling laws \cite{hassani2018almost}.
 While it was proved earlier that RM codes achieve the Shannon capacity of binary erasure channels (BECs) \cite{kudekar2017reed}, and that of binary symmetric channels (BSCs) at extreme rates converging to zero or one \cite{abbe2015reed},  Reeves and Pfister have shown very recently that RM codes are able to achieve the capacity of general binary-input memoryless symmetric (BMS) channels \cite{reeves2021reed}.

Although RM codes have shown excellent performance under maximum likelihood (ML) decoding, they still lack efficient decoding algorithms for general code parameters.
To this end, Dumer’s recursive list decoding algorithm \cite{dumer2006soft}
 provides a complexity-performance trade-off by achieving close-to-ML decoding performance for large enough, e.g., exponential in blocklength, list sizes. 
 Recently, a recursive projection-aggregation (RPA) algorithm was proposed in \cite{ye2020recursive} for decoding RM codes. Despite its explicit structure and excellent decoding performance, the RPA algorithm (in its general form) requires a complexity of $\mathcal{O}(n^r\log n)$ for an RM code of length $n$ and order $r$. Building upon the projection pruning idea in \cite{ye2020recursive}, there has been some recent attempts at reducing the complexity of the RPA algorithm \cite{jamali2021Reed,fathollahi2020sparse}, and also applying it in other contexts than communication \cite{soleymani2021coded}. Moreover, building upon the computational tree of RM (and polar) codes, a class of neural encoders and decoders has been proposed in \cite{makkuva2021ko} via deep learning methods.
 
In this paper, our goal is to devise an efficient, low-complexity, and low-latency coding scheme for low-capacity channels \cite{fereydounian2019channel,jamali2021massive,jamali2018low,dumer2020codes,dumer2021codes},
that are relevant to emerging low-rate communication scenarios, such as narrowband Internet-of-Things (NB-IoT) \cite{ratasuk2016overview}, deep-space communication, and covert (millimeter-wave) communication \cite{jamali2021covert}, among others. Users in these applications typically experience very low signal-to-noise ratios (SNRs). Consequently, reliable communication in such applications requires very large blocklengths, and challenges such as ensuring low latency/complexity and high reliability become more apparent.  
The current practical approaches for these scenarios are mainly based on large repetitions of a powerful moderate-rate code.
While such a construction, i.e., concatenation of a repetition code and a moderate-rate code, results in low-latency codes, it has been shown in \cite{fereydounian2019channel} that the error performance can be significantly degraded 
as a result of 
repetitions. Therefore, using more principled coding schemes  to design low-rate codes can potentially lead to significantly more powerful codes. We will employ the recent advances in RM codes as well as product codes to design efficient coding schemes that achieve  better performance while maintaining  low complexity and low latency. Consequently, our proposed schemes are also of particular application to ultra-reliable and low-latency communications (URLLC).

We build upon product codes \cite{elias1954error} to construct a larger RM code based on the product of smaller RM code components. It is well known that building larger codes upon product codes renders several advantages, such as low encoding and decoding complexity, large minimum distances, and a highly parallelized implementation \cite{elias1954error,pyndiah1998near,mukhtar2016turbo}, and it has very recently been shown that it also enables training neural encoders and decoders for relatively large channel codes \cite{jamali2021productae}.

 \begin{figure*}[t]
	\centering
	\resizebox{0.9\linewidth}{!}{
		\begin{tikzpicture}
		\node at (0.75,0.35) (U) {$\mathbf{U}=\begin{bmatrix} u_{ij} \end{bmatrix}_{k_2\times k_1}$};
		\node at (0.75,0) (P1) {};
		
		
		\coordinate (dm11) at (3,0);
		\coordinate (dm12) at (3,0);
		\node[block,  thick, minimum width=1.5cm, minimum height=1.25cm, align=center, fill=orange!25, draw] [fit = (dm11) (dm12)] (E1) {};
		\node[align=center] at (E1.center) {$\mathcal{E}_1$ \\ $(k_1,n_1)$};
		
		\node at (4.42,0.35) (U1) {$\mathbf{U}^{(1)}_{k_2\times n_1}$};
		
		\coordinate (dm21) at (6,0);
		\coordinate (dm22) at (6,0);
		\node[block, thick,minimum width=1.5cm, minimum height=1.25cm, align=center, fill=orange!25, draw] [fit = (dm21) (dm22)] (E2) {};
		\node[align=center] at (E2.center) {$\mathcal{E}_2$ \\ $(k_2,n_2)$};
		
		\coordinate (dm31) at (9.5,0);
		\coordinate (dm32) at (9.5,0);
		\node[block, thick, minimum width=2.5cm, minimum height=1.25cm, align=center, fill=red!25, draw] [fit = (dm31) (dm32)] (Ch) {};
		\node[align=center] at (Ch.center) {Channel};
		\node at (7.44,0.35) (U2) {$\mathbf{U}^{(2)}_{n_2\times n_1}$};

		\coordinate (dm41) at (12.85,0);
		\coordinate (dm42) at (12.85,0);
		\node[block, thick, minimum width=1.25cm, minimum height=1.25cm, align=center, fill=MyBlue1!25, draw] [fit = (dm41) (dm42)] (D2) {};
		\node[align=center] at (D2.center) {$\mathcal{D}_1$};
		\node at (11.44,0.3) (Y) {$\mathbf{Y}_{n_2\times n_1}$};

		\coordinate (dm51) at (15.6,0);
		\coordinate (dm52) at (15.6,0);
		\node[block,  thick, minimum width=1.25cm, minimum height=1.25cm, align=center, fill=MyBlue1!25, draw] [fit = (dm51) (dm52)] (D1) {};
		\node[align=center] at (D1.center) {$\mathcal{D}_2$};
		\node at (14.15,0.35) (Y2) {$\mathbf{Y}^{(1)}_{n_2\times n_1}$};
		
		\node at (16.9,0.35) (Y1) {$\mathbf{Y}^{(2)}_{n_2\times n_1}$};	
		\node at (17.75,0) (P2) {};

		\draw [->, thick] (P1)--(E1);
		\draw [->, thick] (E1)--(E2);
		\draw [->, thick] (E2)--(Ch);
		\draw [->, thick] (Ch)--(D2);
		\draw [->, thick] (D2)--(D1);
		\draw [->, thick] (D1)--(P2);

		
		\node at (1.15,-0.5) (n1) {};
		\node at (2.25,-0.5) (n2) {};
		\node at (1.15,-0.6) (n3) {};
		\node at (2.25,-0.6) (n4) {};
		\node at (1.15,-0.7) (n5) {};
		\node at (2.25,-0.7) (n6) {};
		\draw [-,densely dashed,thick] (n1)--(n2);
		\draw [-,densely dashed,thick] (n3)--(n4);
		\draw [-,densely dashed,thick] (n5)--(n6);

		\node at (11.15,-0.5) (n7) {};
		\node at (12.25,-0.5) (n8) {};
		\node at (11.15,-0.6) (n9) {};
		\node at (12.25,-0.6) (n10) {};
		\node at (11.15,-0.7) (n11) {};
		\node at (12.25,-0.7) (n12) {};
		\draw [-,densely dashed,thick] (n7)--(n8);
		\draw [-,densely dashed,thick] (n9)--(n10);
		\draw [-,densely dashed,thick] (n11)--(n12);
		
		\node at (4.8,-0.1) (c1) {};
		\node at (4.8,-1.2) (c2) {};
		\draw [-,densely dashed,thick] (c1)--(c2);
		\node at (4.9,-0.1) (c3) {};
		\node at (4.9,-1.2) (c4) {};
		\draw [-,densely dashed,thick] (c3)--(c4);
		\node at (5,-0.1) (c5) {};
		\node at (5,-1.2) (c6) {};
		\draw [-,densely dashed,thick] (c5)--(c6);
		
		\node at (14.55,-0.1) (c7) {};
		\node at (14.55,-1.2) (c8) {};
		\draw [-,densely dashed,thick] (c7)--(c8);
		\node at (14.65,-0.1) (c9) {};
		\node at (14.65,-1.2) (c10) {};
		\draw [-,densely dashed,thick] (c9)--(c10);
		\node at (14.75,-0.1) (c11) {};
		\node at (14.75,-1.2) (c12) {};
		\draw [-,densely dashed,thick] (c11)--(c12);
		
		\end{tikzpicture}
	}
	\caption{
		Demonstration of two-dimensional (2D) product codes. Each $q$-th encoder $\mathcal{E}_q$ and decoder $\mathcal{D}_q$, $q=1,2$, performs encodings and decodings over the $q$-th dimension of the 2D input arrays.}
	\label{fig1}
\end{figure*}
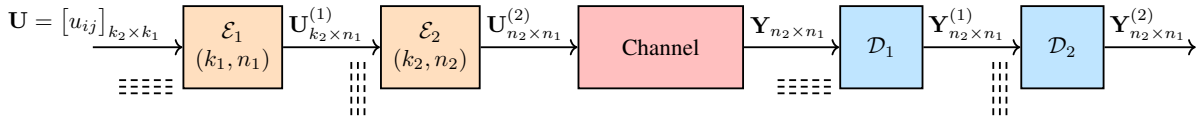
 
 While the framework in this paper is applicable to any RM code components, we particularly consider first-order RM codes as the components to take advantage of their ML performance with an $\mathcal{O}(n\log n)$ complexity, enabled by the fast Hadamard transform (FHT) \cite{ye2020recursive}. The resulting code will be a subcode of an order-$Q$ RM code, when considering $Q$ component codes in the product; thus, it can be a low-rate code depending on the blocklength of individual code components. We present an efficient soft-input soft-output (SISO) iterative decoding algorithm, enabled by our soft-FHT algorithm over code components.

 We show that our decoder maintains a low complexity of $\mathcal{O}(n\log n)$ and a low latency of $\mathcal{O}(\log n)$, regardless of the value of $Q$. Moreover, our numerical results demonstrate the superiority of the proposed SISO decoder compared to hard decoding over RM code components as well as RPA-like decoding of RM subcodes \cite{jamali2021Reed}. We also demonstrate meaningful gains compared to conventional designs such as Turbo-repetition. 
 Lastly, we remark that the proposed methods in this paper can lead to a general framework toward low-complexity decoding of RM codes. 
 

\section{Preliminaries and Setting}\label{Prelim}
\subsection{RM Codes}\label{RM_rev}
An RM code is defined in terms of two parameters: $(i)$ a positive integer $m$ that defines the blocklength as $n=2^m$; and $(ii)$ a nonnegative integer $r\in\{0,1,\cdots, m\}$, named the order of the RM code, that defines the code dimension $k$ as $k=\sum_{i=0}^r \binom{m}{i}$. There are several ways, including the algebraic formulations in \cite{ye2020recursive}, to describe an RM code of length $n=2^m$ and order $r$, denoted by $\mathcal{RM}(m,r)$. One simple description is through the so-called polarization matrix. Indeed, the generator matrix of an $\mathcal{RM}(m,r)$ code, denoted by $\mathbf{G}_{k\times n}$, can be obtained by choosing rows of the following matrix that have a Hamming weight of at least $2^{m-r}$:
\begin{align}\label{Pnn}
\mathbf{P}_{n\times n}=\begin{bmatrix}
1 & 0\\ 1&1
\end{bmatrix}^{\otimes m},
\end{align}
where $\mathbf{F}^{\otimes m}$ is the $m$-th Kronecker power of a matrix $\mathbf{F}$. The resulting generator matrix $\mathbf{G}_{k\times n}$ can then be partitioned into sub-matrices as
\begin{align}\label{Gkn}
\mathbf{G}_{k\times n}=\begin{bmatrix}
\mathbf{G}_{0}\\
\mathbf{G}_{1}\\
\vdots\\
\mathbf{G}_{r}
\end{bmatrix},
\end{align}
where $\mathbf{G}_{0}$ is a length-$n$ all-one row vector, and $\mathbf{G}_{1}$ is an $m\times n$ matrix that lists all the $n=2^m$ unique length-$m$ binary vectors $\{0,1\}^m$ as the columns. Moreover, $\mathbf{G}_{i}$, for $1\leq i\leq r$, is an $\binom{m}{i}\times n$ matrix whose each row is obtained by the element-wise product of a distinct set of $i$ rows from $\mathbf{G}_{1}$ \cite{salomon2005augmented}. Accordingly, $\mathbf{G}_{k\times n}$ has exactly $\binom{m}{i}$ rows with the Hamming weight $n/2^i$, for $0\leq i\leq r$.

\subsection{Product Codes}\label{PCs}
Fig. \ref{fig1} illustrates the encoding and decoding procedure for two-dimensional (2D) product codes. 
Assuming two code components $\mathcal{C}_1:(k_1,n_1)$ and $\mathcal{C}_2:(k_2,n_2)$, their product code is constructed by first forming the length-$k_1 k_2$ information sequence as a $k_2\times k_1$ matrix, and then encoding each row using the first encoder $\mathcal{E}_1$ and each column using the second encoder $\mathcal{E}_2$. It can be shown that in the resulting encoded matrix of size $n_2\times n_1$ (that can be reshaped to a length-$n_1n_2$ vector as a codeword), each row is a codeword of $\mathcal{C}_1$ and each column is a codeword of $\mathcal{C}_2$. Therefore, after properly reshaping the noisy codewords at the receiver, the first decoder $\mathcal{D}_1$ decodes the rows of the received 2D array and the second decoder $\mathcal{D}_2$ decodes the columns of its input array. Note that the order of decoders as well as encoders can be interchanged given the symmetry of the problem. 

In general, a $Q$-dimensional product code $\mathcal{C}$ can be constructed by iterating $Q$ codes $\mathcal{C}_1,\mathcal{C}_2,\cdots,\mathcal{C}_Q$. More specifically, each $q$-th encoder, $q=1,\cdots Q$, encodes the vectors in the $q$-th dimension of the $Q$-dimensional input array.
Similarly, after properly reshaping the noisy codewords at the receiver, each $q$-th decoder decodes the noisy vectors on the $q$-th dimension of the incoming array. Then, assuming $\mathcal{C}_q:(k_q,n_q,d_q,R_q)$ with the generator matrix $\mathbf{G}^{(q)}$, where $d$ and $R$ stand for the minimum distance and rate, respectively, the parameters of the resulting product code $\mathcal{C}$ can be obtained as the product of the parameters of the component codes, i.e.,
\begin{align}
p&=\prod_{q=1}^{Q}p_q,\hspace{1cm} p\in\{k,n,d,R\},\label{P}\\
\mathbf{G}&=\mathbf{G}^{(1)}\otimes \mathbf{G}^{(2)}\otimes\cdots\otimes\mathbf{G}^{(Q)}.\label{G}
\end{align}
It is known that applying a few decoding iterations (together with SISO decoding) usually  improves the decoding performance of product codes \cite{pyndiah1998near}. Therefore, often a few, say $I$, iterations will be applied at the decoder of product codes.

In the special case of RM component codes, the resulting product code is a subcode of a larger RM code, i.e., \cite[Corollary 2]{salomon2005augmented}
\begin{align}\label{subcode}
\mathcal{RM}(m_1,r_1)\otimes\mathcal{RM}(m_2,r_2)\otimes\cdots\otimes\mathcal{RM}(m_Q,r_Q)\nonumber\\
\subseteq \mathcal{RM}\bigg(\sum_{q=1}^{Q}m_q,\sum_{q=1}^{Q}r_q\bigg).
\end{align}
Note, based on \eqref{P}, that the resulting product code has a blocklength of $n_t:=2^{m_t}$, where $m_t:=\sum_{q=1}^{Q}m_q$, that is the same as the blocklength of the larger code in the right-hand side (RHS) of \eqref{subcode}. Also, given that an $\mathcal{RM}(m,r)$ code has a minimum distance of $d=2^{m-r}$, one can observe that both the resulting product code and the code in the RHS of \eqref{subcode} have the same minimum distance $d_t:=2^{m_t-r_t}$, where $r_t:=\sum_{q=1}^{Q}r_q$. However, the resulting product code has a smaller dimension than the larger RM code, i.e.,
\begin{align}\label{k}
\prod_{q=1}^Q\left[\sum_{i_l=0}^{r_q} \binom{m_q}{i_l}\right]\leq \sum_{i_t=0}^{r_t} \binom{m_t}{i_t}.
\end{align}
\subsection{Problem Setting}\label{setting}
In this paper, we consider binary phase-shift keying (BPSK) modulation and transmission over additive white Gaussian noise (AWGN) channels. More specifically, we first map each codeword ${\mathbf{c}}$ to $\tilde{\mathbf{c}}:=1-2{\mathbf{c}}$, before sending it through the channel. The received vector at the channel output is $\mathbf{y}=\tilde{\mathbf{c}}+\mathbf{n}$, where $\mathbf{n}$ is the noise vector whose elements are zero-mean Gaussian random variables with variance $\sigma^2$. In this case, the log-likelihood ratio (LLR) vector can be obtained from $\mathbf{y}$ as $\boldsymbol{l}=2\mathbf{y}/\sigma^2$. Throughout the paper, we define the SNR as ${\rm SNR}:=1/(2\sigma^2)$ and the energy-per-bit $E_b$ to the noise ratio as $E_b/N_0:={\rm SNR}/R=n/(2k\sigma^2)$.

\section{Proposed Scheme}\label{proposed}
\subsection{Encoding Scheme}\label{enc}
The general encoding procedure has been described in Section \ref{PCs}. In this paper, we focus on first-order RM code components with two major motivations. First, using \eqref{subcode}, the resulting product code is a subcode of an $\mathcal{RM}(m_t,Q)$ code, which is a low-rate code for large enough $m_t$'s (compared to $Q$). Therefore, it aligns with the general objective of the paper, which is to design an efficient, low-complexity, and low-latency coding scheme for emerging low-capacity channels. Second, we can take advantage of the low-complexity FHT decoder for order-$1$ RM codes, that achieves the same performance as an ML decoder but with an $\mathcal{O}(n\log n)$ complexity instead of an $\mathcal{O}(n^2)$ complexity. In fact, we establish in Section \ref{complexity} the possibility of decoding the product of any $Q$ first-order RM codes with $\mathcal{O}(n\log n)$ complexity and $\mathcal{O}(\log n)$ latency.

\subsection{Decoding Scheme}\label{dec}
It is not hard to show that the vectors on each $q$-th dimension of the encoded $Q$-dimensional array, at the output of the product encoder, are codewords of the $q$-th component code $\mathcal{C}_q$, even if systematic encoders are not used. Therefore, the vectors on each $q$-th dimension of the received multi-dimensional array, after carefully reshaping the received signal, can be viewed as the noisy codewords of $\mathcal{C}_q$. Accordingly, the decoding procedure can be summarized as Algorithm \ref{alg_dec}. The reshaping of length-$n_t$ vectors to $Q$-dimensional arrays and vice versa, performed in lines 2 and 8, respectively, need to be handled carefully with respect to the product encoder architecture (e.g., the parameter of the individual code components, order of the encoders, etc.). Additionally, in line 5, we considered a general decoder $\mathcal{D}_q$ for decoding the noisy codewords of $\mathcal{C}_q$ on the $q$-th dimension of the LLR array $\mathbf{L}$. In the case of order-$1$ RM codes, considered in this paper as the component codes, we apply a soft version of the FHT algorithm, developed in Section \ref{soft-fht}, to enable an efficient SISO decoding for the underlying product code.

\begin{algorithm}[t]
	\caption{Decoding of $Q$-Dimensional Product Codes}    \label{alg_dec}
	\textbf{Input:} Noisy codeword $\mathbf{y}$, noise variance $\sigma^2$, number of decoding iterations $I$\\
	\textbf{Output:} Decoded codeword $\hat{\mathbf{c}}$
	\vspace*{0.05in}
	\begin{algorithmic}[1]
		\State $\boldsymbol{l}\gets2\mathbf{y}/\sigma^2$ \Comment compute the LLR vector
		\State Properly reshape $\boldsymbol{l}$ to a $Q$-dimensional array $\mathbf{L}$
		\For {$i'=1,2,\cdots,I$}
		\For {$q=1,2,\cdots,Q$}
		\State $\mathbf{L}\gets\mathcal{D}_q(\mathbf{L},{\rm dim}=q)$ \Comment update the vectors on the $q$-th dimension of $\mathbf{L}$ after decoding them using $\mathcal{D}_q$
		\EndFor 
		\EndFor
		\State Properly reshape $\mathbf{L}$ to a length-$n_t$ vector $\hat{\boldsymbol{l}}$
		\State $\hat{\mathbf{c}}\gets0.5(1-{\rm sign}(\hat{\boldsymbol{l}}))$
		\State \textbf{return} $\hat{\mathbf{c}}$
	\end{algorithmic}
\end{algorithm}

\subsection{Soft-FHT Algorithm}\label{soft-fht}
Given $\boldsymbol{l} \in \mathbb{R}^n$ as the vector of channel LLRs, corresponding to the transmission of an $(n,k)$ code with codebook $\mathcal{C}$ over a general binary-input memoryless channel, the ML decoder picks a codeword $\mathbf{c}^{*}$ according to the following rule \cite{ye2020recursive}
\begin{align}\label{map}
\mathbf{c}^{*}=\operatorname*{argmax}_{\mathbf{c}\in\mathcal{C}}~~ \langle\boldsymbol{l},1-2{\mathbf{c}}\rangle,
\end{align}
where $\langle \cdot,\cdot\rangle$ denotes the inner-product of two vectors. A naive implementation of the ML decoder then requires an $\mathcal{O}(n2^k)$ complexity to compute  $2^k$ inner-products between length-$n$ vectors. In particular, for first-order RM codes, $\mathcal{RM}(m,1)$, that have $2^{m+1}=2n$ codewords, this is equivalent to an $\mathcal{O}(n^2)$ complexity and an $\mathcal{O}(n)$ latency (when computing all the inner-products in parallel).
However, one can do the ML decoding for order-$1$ codes in a more efficient way via the FHT algorithm. 
The high-level idea is that half of the $2n$ codewords of an ${\mathcal{RM}}(m,1)$ code (in $\pm 1$) are the columns of the standard $n\times n$ Hadamard matrix $\mathbf{H}$, and the other half are columns of $-\mathbf{H}$. Therefore, the ML decoder for order-$1$ RM codes boils down to the matrix multiplication of the LLR vector $\boldsymbol{l}$ and the Hadamard matrix $\mathbf{H}$, i.e., $\boldsymbol{l}_{\rm WH}:=\boldsymbol{l}\mathbf{H}$, which can be performed in $\mathcal{O}(n\log n)$ complexity and $\mathcal{O}(\log n)$ latency via the FHT algorithm (see \Lref{fht}). Since $\boldsymbol{l}_{\rm WH}$ contains half of the $2n$ inner-products in \eqref{map}, and the other half are just the elements of $-\boldsymbol{l}_{\rm WH}$, the FHT version of the ML decoder for first-order RM codes can be obtained as
\begin{align}\label{map_FHT}
\mathbf{c}^{*}=\frac{1}{2}[1-{\rm sign}(\boldsymbol{l}_{\rm WH}(i^*))\mathbf{h}_{i^*}]~~\text{s.t.}~~i^*=\operatorname*{argmax}_{i=1,2,\cdots n}~\!\!|\boldsymbol{l}_{\rm WH}(i)|,
\end{align}
where $\boldsymbol{l}_{\rm WH}(i)$ is the $i$-th element of the vector $\boldsymbol{l}_{\rm WH}$, and $\mathbf{h}_{i}$ is the $i$-th column of the matrix $\mathbf{H}$.

It will be shown in Section \ref{numerics} that soft decoding of the RM product codes results in a much better performance than their hard decoding. To enable a SISO decoder for RM product codes under consideration, we derive the soft version of the FHT algorithm, referred to as soft-FHT in this paper, for first-order RM code components. We do this in two steps, i.e., first calculating the LLRs of the information bits and then calculating the LLRs of the encoded bits, which will be discussed in the following.

For the AWGN channel model $\mathbf{y}=\tilde{\mathbf{c}}+\mathbf{n}$ and any $(n,k)$ binary linear code $\mathcal{C}$, the LLR $\boldsymbol{l}_{\rm inf}(i)$ of each $i$-th information bit $u_i$, $i=1,2,\cdots,k$, can be obtained form the channel LLRs vector $\boldsymbol{l}$, using the \textit{max-log} approximation, as \cite{jamali2021Reed}
\begin{align}\label{llrinf}
\boldsymbol{l}_{\rm inf}(i) \approx \operatorname*{max}_{\mathbf{c}\in\mathcal{C}_i^0}~\langle \boldsymbol{l}, 1-2{\mathbf{c}}\rangle~-~ \operatorname*{max}_{\mathbf{c}\in\mathcal{C}_i^1}~\langle \boldsymbol{l}, 1-2{\mathbf{c}}\rangle,
\end{align}
where $\mathcal{C}_i^0$ and $\mathcal{C}_i^1$ denote the subsets of codewords that have $u_i=0$ and $u_i=1$, respectively. In the particular case of order-$1$ codes, one can compute $\boldsymbol{l}_{\rm inf}$ more efficiently by invoking the FHT algorithm. 

The generator matrix $\mathbf{G}_{k\times n}$ of a first-order RM code has one row of Hamming weight $n$ and $m$ rows of weight $n/2$. Assuming that the first row is the all-one row, the calculation of $\boldsymbol{l}_{\rm inf}$ for $u_1$ should be carried out differently from the other $u_i$'s. 
Let $\mathbf{U}_{2^k\times k}$ be a matrix listing all binary vectors of length $k$ as the rows such that the $j$-th row, $j=1,2,\cdots 2^k$, is the binary representation of the number $j-1$ in $k$ bits with the most significant bit being at the left. The matrix multiplication $\mathbf{C}_{2^k\times n}:=\mathbf{U}\mathbf{G}$ (over the binary field) then lists all the codewords in a way that the upper half (the first $n$ rows) of $\tilde{\mathbf{C}}:=1-2\mathbf{C}$ is equal to $\mathbf{H}$ and the lower half is equal to $-\mathbf{H}$. Therefore, given that $u_1$ is equal to zero for the first half of the codewords and equal to one for the second half, we have using \eqref{llrinf}
\begin{align}\label{llrinf1}
\boldsymbol{l}_{\rm inf}(1) \approx \operatorname*{max}_{i'=1,2,\cdots n}~\boldsymbol{l}_{\rm WH}(i')~-~ \operatorname*{max}_{i'=1,2,\cdots n}~-\boldsymbol{l}_{\rm WH}(i').
\end{align}

To compute the LLRs $\boldsymbol{l}_{\rm inf}(i)$ for $i=2,\cdots k$, we only need to find the set of indices of the first half of the codewords that have $u_i=0$ and $u_i=1$, denoted by the sets $\mathcal{I}_{0,i}\subset\{1,2,\cdots n\}$ and $\mathcal{I}_{1,i}\subset\{1,2,\cdots n\}$, respectively\footnote{Note that these sets of indices are fixed across the decoding and can be computed before hand to reduce the decoding complexity and latency.}. In fact, for any codeword in the first half that has $u_i=0$ or $u_i=1$, we have exactly the negative of that codeword in the second half, corresponding to the same realization of the bits $(u_1,u_2,\cdots, u_k)$ but with $u_1=1$ instead of $u_1=0$ (recall that the first row of $\mathbf{G}$ is all-one). Therefore, using \eqref{llrinf}, we have
\begin{align}\label{llrinfk}
\boldsymbol{l}_{\rm inf}(i\neq1) &~\approx \operatorname*{max}_{i'\in\mathcal{I}_{0,i}}~\pm\boldsymbol{l}_{\rm WH}(i')~-~ \operatorname*{max}_{i'\in\mathcal{I}_{1,i}}~\pm\boldsymbol{l}_{\rm WH}(i')\nonumber\\
&~= \operatorname*{max}_{i'\in\mathcal{I}_{0,i}}~|\boldsymbol{l}_{\rm WH}(i')|~-~ \operatorname*{max}_{i'\in\mathcal{I}_{1,i}}~|\boldsymbol{l}_{\rm WH}(i')|.
\end{align}

\begin{algorithm}[t]
	\caption{Soft-FHT Algorithm for $\mathcal{RM}(m,1)$ Codes} \label{alg_softfht}
	\textbf{Input:} The channel LLR vector $\boldsymbol{l}$; RM code parameter $m$,  the sets of indices $\mathcal{I}_{0,i}$ and $\mathcal{I}_{1,i}$ for each $i$-th bit, $i=2,\cdots m+1$\\
	\textbf{Output:} Soft decisions (i.e., the updated LLR vector) $\hat{\boldsymbol{l}}$
	\vspace*{0.05in}
	\begin{algorithmic}[1]
		\State $\boldsymbol{l}_{\rm WH}\gets\boldsymbol{l}\mathbf{H}$ \Comment apply FHT algorithm to $\boldsymbol{l}$
		\State  Initialize $\boldsymbol{l}_{\rm inf}$ as an all-zero vector of length $m+1$
		\State $\boldsymbol{l}_{\rm inf}(1) \gets$ Eq. \eqref{llrinf1} \Comment calculate $\boldsymbol{l}_{\rm inf}(1)$ using \eqref{llrinf1}
		\For {$i=2,\cdots,m+1$} 
		\State $\boldsymbol{l}_{\rm inf}(i) \gets$ Eq. \eqref{llrinfk} \Comment calculate $\boldsymbol{l}_{\rm inf}(i)$ using \eqref{llrinfk}
		\EndFor 
		\State  Initialize $\boldsymbol{l}_{\rm enc}$ as an all-zero vector of length $n:=2^m$
		\State $\mathbf{R}\gets \texttt{repeat}(\boldsymbol{l}_{\rm inf}^T,1,n)$ \Comment  concatenate $n$ copies of $\boldsymbol{l}_{\rm inf}^T$
		\State $\mathbf{V}\gets \mathbf{R}\odot \mathbf{G}$ \Comment element-wise matrix multiplication
		\For {$j=1,2,\cdots,n$}
		\State $\mathbf{v}\gets$ nonzero elements in the $j$-th column of $\mathbf{V}$
		\State $\boldsymbol{l}_{\rm enc}(j)\!\gets\!\prod_{j'}{\rm sign}(\mathbf{v}(j')) \!\times\! \operatorname*{min}_{j'} \!|\mathbf{v}(j')|$ \Comment using \eqref{llr_enc}
		\EndFor 
		\State $\hat{\boldsymbol{l}}\gets\boldsymbol{l}_{\rm enc}$
		\State \textbf{return} $\hat{\boldsymbol{l}}$
	\end{algorithmic}
\end{algorithm}

Once we have the LLRs of the information bits, we can use them to calculate the LLRs of the encoded bits, denoted by $\boldsymbol{l}_{\rm enc}$. Note that the $j$-th encoded bit $c_j$, $j=1,\cdots,n$, is obtained using the $j$-th column of $\mathbf{G}$ as $c_j=\sum_{i=1}^{m+1}u_ig_{i,j}$. Therefore, the LLR $\boldsymbol{l}_{\rm enc}(j)$ of the $j$-th encoded bit can be obtained using the well-known \textit{min-sum} approximation as
\begin{align}\label{llr_enc}
\boldsymbol{l}_{\rm enc}(j)=\prod_{i\in \Lambda_{j}}{\rm sign}(\boldsymbol{l}_{\rm inf}(i)) \times \operatorname*{min}_{i\in \Lambda_{j}} |\boldsymbol{l}_{\rm inf}(i)|,
\end{align}
where $\Lambda_{j}$ is the set of indices corresponding to the nonzero elements in the $j$-th column of $\mathbf{G}$. The soft-FHT algorithm is summarized in Algorithm \ref{alg_softfht}.

\subsection{Complexity and Latency Analysis}\label{complexity}
The following two lemmas establish sufficient conditions for decoding \textit{any} $Q$-dimensional product code with an $\mathcal{O}(n\log n)$ complexity and an $\mathcal{O}(\log n)$ latency.
\begin{lemma}\label{low-comp}
	Any $Q$-dimensional product code can be decoded with an $\mathcal{O}(n\log n)$ complexity if the component codes can be decoded with an $\mathcal{O}(n\log n)$ complexity.
\end{lemma}
\begin{proof}
	Let $\mathcal{N}(n_q,k_q)$ denote the decoding complexity of the $q$-th decoder, $q=1,2,\cdots,Q$. At each iteration, the decoder needs to perform $n_t/n_q$ decodings over length-$n_q$ vectors, each incurring an $\mathcal{N}(n_q,k_q)$ complexity. Given that there are $Q$ decoders at each iteration, the overall decoding complexity $\mathcal{N}_t$ will be
	\begin{align}\label{com_tot}
	\mathcal{N}_t&=I\sum_{q=1}^{Q}\frac{n_t}{n_q}\mathcal{N}(n_q,k_q)\nonumber\\
	&\overset{(a)}{=}In_t\sum_{q=1}^{Q}\mathcal{O}(\log n_q)\nonumber\\
	&\overset{(b)}{=}In_t\mathcal{O}(\log n_t),
	\end{align}
	where step $(a)$ is by the assumption that the $q$-th decoder requires $\mathcal{N}(n_q,k_q)=\mathcal{O}(n_q\log n_q)$ complexity, and step $(b)$ follows by $\sum_{q=1}^{Q}\log n_q=\log \prod_{q=1}^{Q}n_q=\log n_t$. As we numerically verify in Section \ref{numerics}, $I$ is a small number (usually less than $5$) and does not impact the complexity and latency.
\end{proof}
\begin{lemma}\label{low-lat}
	Any $Q$-dimensional product code can be decoded with an $\mathcal{O}(\log n)$ latency if the component codes can be decoded with an $\mathcal{O}(\log n)$ latency.
\end{lemma}
\begin{proof}
	Given that all $n_t/n_q$ decodings at each $q$-th dimension can be executed in parallel, the overall latency is $I\sum_{q=1}^{Q}\mathcal{O}(\log n_q)=I\mathcal{O}(\log n_t)$.
\end{proof}
\begin{lemma}\label{fht}
	Besides having an $\mathcal{O}(n\log n)$ complexity, the FHT algorithm performs the ML decoding in $\mathcal{O}(\log n)$ latency for first-order RM codes of blocklength $n$. 
\end{lemma}
\begin{proof}
The core idea behind the implementation of the FHT algorithm is that the $2^m\times 2^m$ matrix $\mathbf{H}$ can be written as the product of $m$ matrices of size $2^m\times 2^m$, say $\mathbf{M}_1,\mathbf{M}_2,\cdots,\mathbf{M}_m$, each having only two non-zero elements per column \cite[page 421]{macwilliams1977theory}. Therefore, 
\begin{align}
\boldsymbol{l}_{\rm WH}:=\boldsymbol{l}\mathbf{H}=\boldsymbol{l}\mathbf{M}_1\mathbf{M}_2\cdots\mathbf{M}_m
\end{align}
boils down to $m$ matrix multiplications of the form $\mathbf{f}_{s}:=\mathbf{f}_{s-1}\mathbf{M}_s$, $s=1,2,\cdots,m$, with $\mathbf{f}_{0}:=\boldsymbol{l}$. Given that each matrix $\mathbf{M}_s$ has two non-zero elements per column, we only need a single addition/subtraction to compute each of $2^m$ elements of each vector $\mathbf{f}_{s}$. Therefore, each $\mathbf{f}_{s}$ can be computed with $\mathcal{O}(2^m)$ complexity and $\mathcal{O}(1)$ latency (when computing all $2^m$ elements of $\mathbf{f}_{s}$ in parallel). Finally, since each of $m$ vectors $\mathbf{f}_{s}$'s should be computed serially, to get $\boldsymbol{l}_{\rm WH}$, we need $\mathcal{O}(m2^m)$ complexity and $\mathcal{O}(m)$ latency in total.
\end{proof}
\begin{theorem}\label{rm_comp_lat}
	Any RM subcode that is obtained as the product of order-$1$ RM codes can be decoded in $\mathcal{O}(n\log n)$ complexity and $\mathcal{O}(\log n)$ latency via soft-FHT algorithm over component codes.
\end{theorem}
\begin{proof}
	This follows immediately from Lemmas \ref{low-comp} and \ref{low-lat}, and noting that the proposed soft-FHT algorithm, similar to the FHT algorithm, requires $\mathcal{O}(n\log n)$ complexity and $\mathcal{O}(\log n)$ latency to decode order-$1$ RM codes.
\end{proof}

\begin{theorem}\label{enc_comp_lat}
The proposed coding scheme has the encoding complexity of $\mathcal{O}(n\log n)$ and encoding latency of $\mathcal{O}(\log n)$.
\end{theorem}
\begin{proof}
	Note, based on the general encoding procedure of binary linear codes $\mathbf{c}=\mathbf{u}\mathbf{G}$, that the encoding complexity and latency are $\mathcal{O}(kn)$ and $\mathcal{O}(k)$, respectively. For order-$1$ RM code components we have $k=m+1=1+\log n$, which results in the encoding complexity and latency of $\mathcal{O}(n\log n)$ and $\mathcal{O}(\log n)$, respectively, for the code components. Following similar procedures to Lemmas \ref{low-comp} and \ref{low-lat}, one can show that the overall encoding complexity and latency of any $Q$-dimensional product code are also $\mathcal{O}(n\log n)$ and $\mathcal{O}(\log n)$, respectively, if the underlying code components have that encoding complexity and latency. 
\end{proof}
\section{Numerical Results}\label{numerics}
In this section, we present extensive numerical results to study the performance of the proposed coding scheme in various aspects, while focusing on 2D product codes. We first verify the accuracy of the soft-FHT decoder in Fig. \ref{fig2}. As seen, all decoders, namely FHT, soft-FHT, MAP, and soft-MAP \cite{jamali2021Reed}, match for order-$1$ RM codes. Fig. \ref{fig2} also shows the impact of the number of iterations $I$ on the performance of a sample product code, i.e.,  $\mathcal{RM}(6,1)\otimes\mathcal{RM}(2,1)$. It is observed that not many iterations are required for our proposed decoder.
\begin{figure}[t]
	\centering
	\includegraphics[trim=0.5cm 0.5cm 0.4cm 0.8cm,width=3.6in]{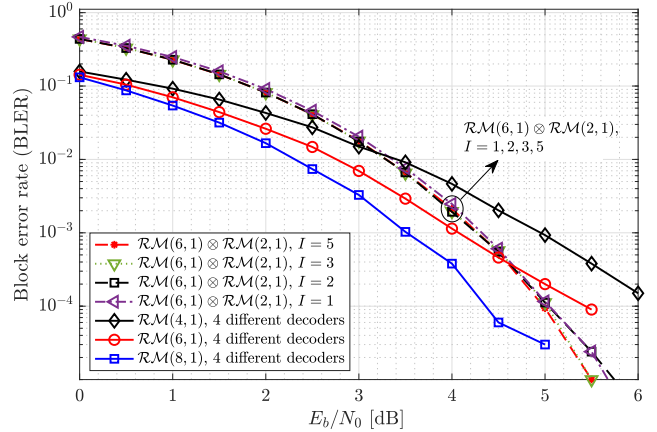}
	\caption{Accordance of the performance of $4$ different decoders, namely FHT, soft-FHT, MAP, and soft-MAP \cite{jamali2021Reed}, for first-order RM codes. The impact of the number of iterations $I$ is also illustrated for $\mathcal{RM}(6,1)\otimes\mathcal{RM}(2,1)$.}
	\label{fig2}
\end{figure}

Note that, as shown in Fig. \ref{fig1}, we first do the decoding over $\mathcal{C}_1$ and then over $\mathcal{C}_2$. As such, the decoder $\mathcal{D}_1$ is expected to decode noisier codewords than $\mathcal{D}_2$. Therefore, one needs to use a stronger code (e.g., with a larger blocklength and/or a lower rate) for $\mathcal{C}_1$ compared to $\mathcal{C}_2$. In the context of the product of order-$1$ RM codes, considered here, this is equivalent to having $m_1>m_2$. This is confirmed in Figs. \ref{fig3} and \ref{fig4} for subcodes of $\mathcal{RM}(13,2)$ and $\mathcal{RM}(8,2)$, obtained as the product of $\mathcal{RM}(m_1,1)\otimes\mathcal{RM}(m_2,1)$ such that $m_1+m_2=13$ and $m_1+m_2=8$, respectively. It is observed that the system performance improves\footnote{Note that the channel capacity is approximately linear in ${\rm SNR}$ over low-capacity regimes. Therefore, based on the definition of $E_b/N_0$, it is logical to compare the performance of different low-rate codes
in terms of $E_b/N_0$.} as we increase $m_1-m_2$. 

Fig. \ref{fig4} also compares the performance of hard decoding with soft decoding for various subcodes of $\mathcal{RM}(8,2)$. The results for hard decoding are obtained by applying the FHT algorithm to the component codes to return hard decisions of the noisy codewords over each dimension. The hard decisions $\hat{\mathbf{y}}_i\in\{0,1\}^{n_i}$, $i=1,2$, are then mapped to $1-2\hat{\mathbf{y}}_i$ before feeding the next FHT decoder. As seen, our SISO decoder significantly outperforms hard decoding. Additionally, the same trend is observed for hard decoding as we increase $m_1-m_2$.
\begin{figure}[t]
	\centering
	\includegraphics[trim=0.5cm 0.6cm 0.4cm 0.8cm,width=3.6in]{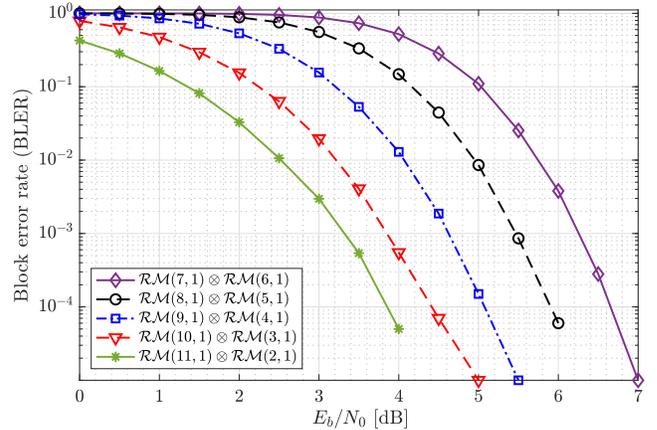}
	\caption{Impact of code component parameters on the performance of various subcodes of $\mathcal{RM}(13,2)$.}
	\label{fig3}
\end{figure}
\begin{figure}[t]
	\centering
	\includegraphics[trim=0.5cm 0.55cm 0.4cm 0.8cm, width=3.6in]{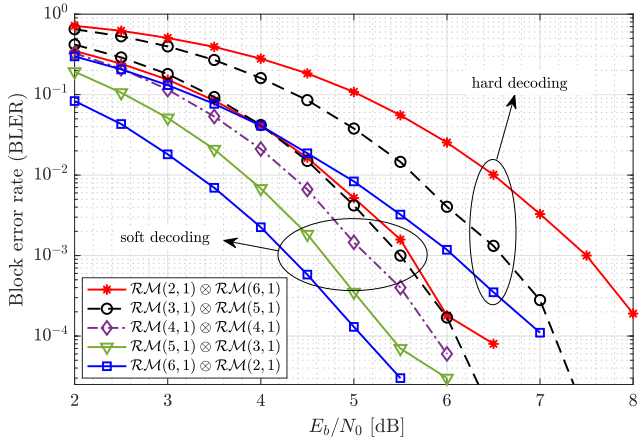}
	\caption{Impact of code component parameters on the performance of various subcodes of $\mathcal{RM}(8,2)$. The comparison between hard decoding and soft decoding is also included.}
	\label{fig4}
\end{figure}

To demonstrates the efficiency of the proposed SISO decoder, we compare its performance with the sub-RPA algorithm \cite{jamali2021Reed}, that achieves close-to-ML performance though with full-projection decoding incurring $\mathcal{O}(n^r\log n)$ complexity for a subcode of $\mathcal{RM}(m,r)$. Fig. \ref{fig5} shows that the full-projection sub-RPA decoding outperforms our low-complexity and low-latency decoder by almost $0.5$ \si{dB} at the BLER of $10^{-3}$, for a subcode of $\mathcal{RM}(8,2)$ obtained as the product of $\mathcal{RM}(6,1)\otimes\mathcal{RM}(2,1)$. However, a more fair comparison is to limit the number of projections in the sub-RPA decoder to a level with a comparable complexity to our SISO decoder. Indeed, the full-projection sub-RPA decoder applies $n-1=255$ projections resulting in $\mathcal{O}(n^2\log n)$ overall complexity. If we apply $5$ random projections for the sub-RPA decoder (we tried $8$ different random selections of $5$ subspaces from $255$ possible subspaces), the performance is then inferior to our SISO decoder by a large margin. Also, the sub-RPA algorithm cannot beat our low-complexity decoder even with $16$ projections (that is still much more complex than our decoder). Our additional simulations with $32$ projections for the sub-RPA decoder show that there are a few ($2$ out of $8$) random trials of the selection of projections that get close to our decoder, while most of the random trials with $64$ projections get slightly better than our decoder.

\begin{figure}[t]
	\centering
	\includegraphics[trim=0.5cm 0.5cm 0.4cm 0.8cm,width=3.6in]{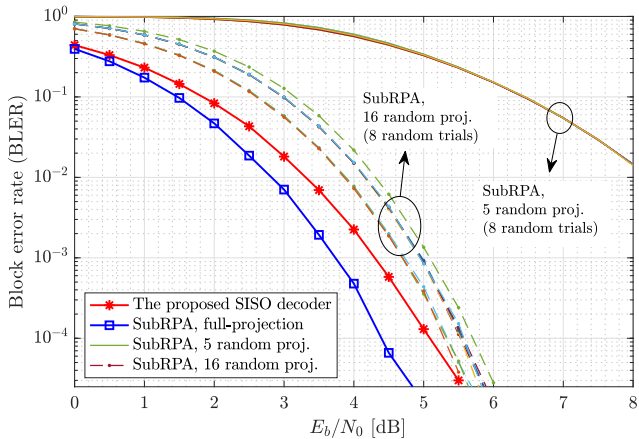}
	\caption{Comparison of the proposed SISO decoder with the sub-RPA algorithm \cite{jamali2021Reed} with full-projection as well as $5$ and $16$ random projections. Product code $\mathcal{RM}(6,1)\otimes\mathcal{RM}(2,1)$ is considered.}
	\label{fig5}
\end{figure}

\begin{figure}[t]
	\centering
	\includegraphics[trim=0.5cm 0.55cm 0.4cm 0.8cm,width=3.6in]{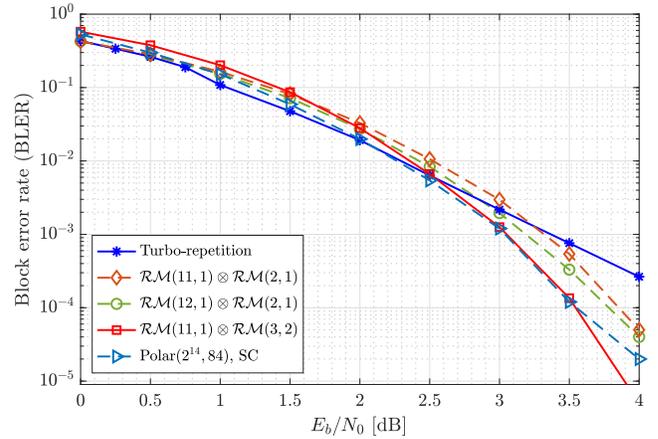}
	\caption{Comparison of the proposed coding scheme with Turbo-repetition and polar under successive cancellation (SC) decoding.}
	\label{fig6}
\end{figure}

Finally, Fig. \ref{fig6} compares the performance of the proposed coding scheme with Turbo-repetition and polar codes. the Turbo-repetition is obtained by repeating a $(120,40)$ Turbo code $68$ times to obtain a $(40,8160)$ code. It is observed that the equivalent RM product codes have sharper slopes and achieve much better performances over moderate to low BLER regimes, thus demonstrating potential applications to URLLC. Fig. \ref{fig6} also shows that it is useful to increase the rate of the second component when the first component is a strong enough code to support such a high rate. For example, $\mathcal{RM}(11,1)\otimes \mathcal{RM}(3,2)$ (via soft-MAP \cite{jamali2021Reed} over $\mathcal{RM}(3,2)$) achieves almost $0.3$ dB gain over $\mathcal{RM}(12,1)\otimes \mathcal{RM}(2,1)$ and $0.9$ dB over Turbo-repetition at the BLER of $10^{-4}$ (note that the performance of Turbo-repetition does not change in $E_b/N_0$ by doubling the number of repetitions as the SNR will increase by the same factor of two that the rate is decreased). Moreover, our $\mathcal{RM}(11,1)\otimes \mathcal{RM}(3,2)$ code achieves the same performance as the equivalent polar code of parameters $(2^{14},84)$, under successive cancellation (SC) decoding, despite its much lower latency. List decoding of the proposed RM product codes to further improve their performance is a subject of future research.

\section{Conclusions}\label{concl}
In this paper, we presented a low-complexity and low-latency coding scheme, based on the product of smaller (particularly, first-order) RM code components, with particular applications to emerging low-capacity scenarios. We proposed an iterative SISO decoder enabled by soft-FHT decoding of code components. It was shown that the proposed coding scheme requires $\mathcal{O}(n\log n)$ complexity and $\mathcal{O}(\log n)$ latency for both  encoding and decoding. Through extensive numerical results, we studied the performance and efficiency of the proposed coding scheme in various aspects. 
Given the recent breakthrough result in \cite{reeves2021reed} proving the capacity-achievability of RM codes over any BMS channel, the design of efficient decoders for RM codes becomes even more substantial than ever. And, based on the fact that any RM code can be written as the union of RM subcodes defined as the product of smaller RM codes \cite{salomon2005augmented}, we believe that the research in this paper opens a new framework toward efficient decoding of RM codes.

\bibliographystyle{IEEEtran}
\bibliography{IEEEabrv}

\end{document}